\documentclass[11pt]{article}

\usepackage[letterpaper,margin=1in]{geometry}
\usepackage[utf8]{inputenc} % allow utf-8 input
\usepackage[T1]{fontenc}    % use 8-bit T1 fonts
\usepackage[colorlinks=true]{hyperref}       % hyperlinks
\usepackage{url}            % simple URL typesetting
\usepackage{booktabs}       % professional-quality tables
\usepackage{amsfonts}       % blackboard math symbols
\usepackage{amsmath,dsfont,amsthm,mathtools,amssymb}
\usepackage{subcaption,enumitem}

\usepackage{tikz,soul,pgfplots}
\pgfplotsset{compat=1.18}
\usetikzlibrary{shapes,arrows,fit, calc}
\usetikzlibrary{shapes.misc, positioning,patterns}
\usetikzlibrary{decorations.pathreplacing}
\usetikzlibrary{arrows.meta, shapes,patterns.meta}

\usepackage{amsthm}

\DeclareMathOperator{\DSBS}{DSBS}
\DeclareMathOperator{\BSC}{BSC}
\DeclareMathOperator{\var}{var}
\DeclareMathOperator{\rank}{rank}
\DeclareMathOperator{\Bern}{Bernoulli}

\newcommand{\calC}{\mathcal{C}}

\newcommand{\calR}{\mathcal{R}}

\renewcommand{\o}{\mathsf{o}}

\newcommand{\probP}{\mathsf{P}}

\newcommand{\bbE}{\mathbb{E}}

\renewcommand{\d}{\mathrm{d}}

\newcommand{\binent}{h_2}
\newcommand{\bindiv}{d_2}

\newcommand*\markov{\mathrel{-\mkern-3mu{\circ}\mkern-3mu-}}
\newcommand{\ones}{\mathbf{1}}
\newcommand{\zeros}{\mathbf{0}}
\newcommand{\eye}{\mathbf{I}}

\newcommand{\hyp}{\mathcal{H}}
\newcommand{\party}[1]{\mathsf{#1}}
\newcommand{\better}{\succeq_{\text{B}}}
\newcommand{\truncE}{E_{\text{tr}}}
\newcommand{\HanE}{E_{\text{Han}}}

\newcommand{\comE}{E_{\text{com}}}
\newcommand{\Wmat}{W}
\newcommand{\bvec}{\vec{b}}
\newcommand{\avec}{\vec{a}}

\newtheorem{theorem}{Theorem}
\newtheorem{conjecture}{Conjecture}

\newtheorem{fact}{Fact}

\hypersetup{
	colorlinks,
	linkcolor=blue,
	citecolor=red,
	urlcolor=blue
}

\usepackage[
backend=biber,
style=alphabetic,
maxbibnames=15,
maxalphanames=5,
minalphanames=3,
doi=false,
isbn=false,
eprint=false,
backref=true,
]{biblatex}
\title{On the Suboptimality of Linear Codes for Binary Distributed Hypothesis Testing} 

\author{%
	Adway Girish, Robinson D.~H.~Cung, Emre Telatar \\
	{\small    School of Computer and Communication Sciences, EPFL} \\
	{\small   \texttt{\{adway.girish,robinson.cung,emre.telatar\}@epfl.ch} }
}

\begin{document}

	\maketitle
	
	\begin{abstract}
		We study a binary distributed hypothesis testing problem where two agents observe correlated binary vectors and communicate compressed information at the same rate to a central decision maker. In particular, we study linear compression schemes and show that simple truncation is the best linear scheme in two cases: (1) testing opposite signs of the same magnitude of correlation, and (2) testing for or against independence. We conjecture, supported by numerical evidence, that truncation is the best linear code for testing \emph{any} correlations of opposite signs.  Further, for testing against independence, we also compute classical random coding exponents and show that truncation, and consequently any linear code, is strictly suboptimal.
	\end{abstract}

	\section{Introduction}
	\label{sec: intro}
	
	Consider the following distributed hypothesis testing (DHT) setup with two sensors and a central decision maker.  Under hypothesis $\hyp = i \in \{0,1\}$, the pair of random variables $(X,Y)$ has joint distribution $\probP^i_{XY}$.  The two sensors observe $n$ independent copies of $X$ and $Y$ respectively and send a compressed version of these observations to the central decision maker, who then declares $\hat \hyp = 0$ or $1$ according to some decision rule, as shown in Fig.~\ref{fig: dht_fig}.  More precisely, under hypothesis $\hyp = i$, with $(X_\ell, Y_\ell)$, $\ell = 1,\dots,n$ drawn i.i.d.\ from distribution $\probP^i_{XY}$, agent $\party{A}$ observes $X^n$ and agent $\party{B}$ observes $Y^n$.  These agents compress $X^n$ and $Y^n$ using encoding functions $g$ and $h$ respectively, and send $g(X^n)$ and $h(Y^n)$ to the central decision maker $\party{C}$, who then declares $\hat \hyp$.
	
	This is an instance of the more general problem of statistical inference under communication constraints, first studied by Berger~\cite{berger1979decentralized}, see the classical survey by Han and Amari~\cite{han_amari_survey} for an extensive overview of the early results and techniques.  The DHT problem has received renewed interest in recent years owing to such distributed setups arising, for example, in modern remote sensor networks, where raw data cannot be centrally aggregated due to communication, latency, or privacy limitations~\cite{sensor_networks,sensor_new}.  Several variants of the above problem have been studied, such as allowing multiple rounds of interaction~\cite{DHT_interactive}, making decisions at multiple centers~\cite{DHT_cooperation}, adding an explicit privacy metric~\cite{DHT_privacy}, over noisy channels~\cite{DHT_noisy}, and so on.  For the setup in Fig.~\ref{fig: dht_fig}, most choices of $g$ and $h$ proposed in the literature are based on typicality-based quantization and binning arguments, as in Ahlswede--Csisz\'ar~\cite{ahlswede_csiszar}, Han~\cite{han_multiterminal}, Shimokawa--Han--Amari~\cite{SHA}, Kochman--Wang~\cite{kochman_wang_improved}, and Watanabe~\cite{watanabe_binning}, but the optimal choice remains unknown. 
	
	\begin{figure}[!tbhp]
		\centering
		\tikzset{
  block/.style    = {draw, thick, rectangle, minimum height = 1cm,minimum width = 1.5em}
}

\begin{tikzpicture}[auto, node distance=2cm,>=latex', example/.style={rectangle callout, draw=black, block/.style  = {draw, thick, rectangle, minimum width = 1em}}]
    % place nodes
    \node [block, inner sep=1.3em] at (4,0) (dec) {$\party{C}$};
    \node [coordinate, right of=dec, node distance=1.8cm] (output) {};
    \draw [thick,->] (dec) -- node[anchor=west] {\hspace{1.6em}$\hat \hyp$} (output);

    \node at (-1.5,1.3) (inputx) {$X^n$};
    \node at (-1.5,-1.3) (inputy) {$Y^n$};
    \node [block, inner sep=1em] at (0.5,1.3) (encx) {$\party{A}$};    
    \node [block, inner sep=1em] at (0.5,-1.3) (ency) {$\party{B}$};  
    % connect nodes
    \draw [thick,->] (inputx) --  (encx);
    \draw [thick,->] (inputy) --  (ency);
    \draw [thick,densely dashed,->] (encx) -- node[pos=0.2, anchor=south west] {$g(X^n)$} (dec);
    \draw [thick,densely dashed,->] (ency) -- node[pos=0.2, anchor=north west] {$h(Y^n)$} (dec);
\end{tikzpicture}
		\caption{Distributed hypothesis testing setup considered.  The communication from $\party{A}$ and $\party{B}$ to $\party{C}$ is constrained. We take $(X^n, Y^n) \sim \DSBS(p_i)^{\otimes n}$ under hypothesis $\hyp = i \in \{0,1\}$ and study the performance of linear $g,h$.}
		\label{fig: dht_fig}
	\end{figure}
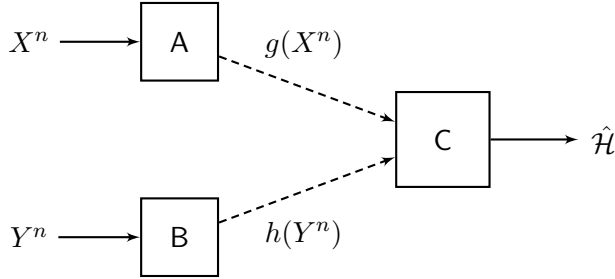
	
	In this paper, we take $X$ and $Y$ to be binary random variables uniformly distributed on $\{0,1\}$, with the hypothesis determining the correlation between $X$ and $Y$.  Explicitly, we have $Y = X \oplus Z$, where $Z \sim \Bern(p_i)$ is independent of $X$  (and also $Y$) under hypothesis $\hyp = i$.  This is exactly the setup considered by Haim and Kochman~\cite{haim_kochman_itw,haim_kochman_arxiv}, except that they restrict themselves to the case with $p_0 \leq p_1 \leq \frac12$ (i.e., $X$ and $Y$ are positively correlated under both hypotheses).  They show, inspired by the K\"orner--Marton scheme~\cite{korner_marton} for distributed computation, that linear codes can be used to obtain performance comparable to the unconstrained/co-located problem (where there is no compression), even obtaining a better Stein exponent than with classical typicality-based random quantization schemes~\cite{ahlswede_csiszar,han_multiterminal}. 
Linear codes are also attractive for being computationally tractable, allowing for practically realizable compression schemes~\cite{practical_DHT}.  
We focus on the case when $p_0$ and $p_1$ are on opposite sides of $\frac12$ (i.e., $X$ and $Y$ have correlations of opposite signs under the two hypotheses), and
claim that linear codes achieve only a trivial performance, i.e., no better than straightforward truncation.  This is similar in spirit to Amari~\cite{amari_fisher}, who showed, by computing the Fisher information of $(g(X^n), h(Y^n))$, that truncation is asymptotically optimal for estimating the correlation when $X,Y$ are nearly independent, among a different class of ``linear-threshold encodings''.
In particular, we show the following:
	\begin{itemize}
		\item[1.] For any values of correlations under the two hypotheses, there is no advantage in allowing the agents to use different linear codes, or equivalently, the best error trade-off is obtained by choosing the same linear code at both agents.
		\item[2.] For (i) testing opposite signs of the same magnitude of correlation, and for (ii) testing for or against independence, straightforward truncation is the best linear code. We conjecture that this holds when the correlations have opposite signs under the two hypotheses.
		\item[3.] For (i) and (ii) above, classical typicality-based quantization achieves a better Stein exponent than truncation, and hence any linear code, making them strictly suboptimal. 
	\end{itemize}

	\section{Preliminaries and Problem Setup} \label{sec: prelims}
	
	To fix ideas, we first describe the general hypothesis testing framework in a one-shot fashion, i.e., with only one sample.  Naturally, this subsumes our setting with $n$ i.i.d.\ samples as the single general sample is allowed to be a vector drawn from a product distribution.
	
	\subsection{General one-shot setup}
	
	Under hypothesis $\hyp = i$, we have $(X,Y) \sim \probP^i_{XY}$.  We may drop the subscripts when the random variables are clear from context.
	
	We first consider a co-located version of the problem:
	suppose the decision maker $\party{C}$ observes both $X$ and $Y$.  The goal of $\party{C}$ is to obtain an estimate $\hat \hyp$ of $\hyp$ from $(X,Y)$, i.e., a \emph{decision rule} $\hat \hyp (X,Y)$. Associated with each decision rule $\hat \hyp$ are two kinds of errors, called the type-I and type-II errors, with probabilities respectively given by
	\begin{align*}
		e_1 &= \Pr\{\hat \hyp = 1 \mid \hyp = 0\} = \probP^0 (\hat \hyp (X,Y) = 1), \\ e_2 &= \Pr\{\hat \hyp = 0 \mid \hyp = 1\} = \probP^1 (\hat \hyp (X,Y) = 0).
	\end{align*}
	Let $\calR(\probP^0_{XY}, \probP^1_{XY}) \subset [0,1]^2$ denote the set of all pairs $(e_1,e_2)$ achieved by some decision rule $\hat \hyp$. An \emph{optimal} decision rule is one that lies on the lower-left boundary of this region.  As is well-known, these optimal decision rules are the randomized Neyman--Pearson likelihood ratio tests---pick a threshold $t \geq 0$, and decide $\hat \hyp = 0$ if $\probP^0(X,Y) / \probP^1(X,Y) > t$, $1$ if $\probP^0(X,Y) / \probP^1(X,Y) < t$, and randomly in case of equality.  Sweeping $t$ from $0$ to $\infty$ gives the lower-left boundary of $\calR(\probP^0_{XY}, \probP^1_{XY})$.
	Hence, the likelihood ratio is a sufficient statistic, i.e., no decision rule computed using $(X,Y)$ can have both type-I and type-II error probabilities lower than the optimal decision rule computed using $\probP^0(X,Y) / \probP^1(X,Y)$.
	
	More generally, let $U = u(X,Y)$ and $V = v(X,Y)$ be two statistics computed from $(X,Y)$ using possibly randomized functions $u$ and $v$ respectively (in all randomized functions we discuss, the randomness cannot depend directly on the hypothesis).  We say that $U$ is \emph{better} than $V$ for testing $(X,Y) \sim \probP^0$ versus $\probP^1$, denoted by $U \better V$, if the optimal decision rule using $U$ has uniformly better type-I and type-II error probabilities than that using $V$.  Clearly, this happens if and only if
	$\calR(\probP^0_{U}, \probP^1_{U}) \supset \calR(\probP^0_{V}, \probP^1_{V})$. 
	We trivially have $(X,Y) \better u(X,Y)$ for any $u$.
	As the likelihood ratio is a sufficient statistic, we also have $\probP^0(X,Y) / \probP^1(X,Y) \better (X,Y)$.
	The following classical result due to Blackwell~\cite{blackwell_order},~\cite[Ex.~III.1]{Polyanskiy_Wu_2025} characterizes exactly when this partial order relation is satisfied.

	\begin{fact}[Blackwell order] Let $T$ have distribution $\probP^i$ under hypothesis $\hyp = i$, and let $U = u(T)$ and $V = v(T)$ be two statistics computed from $T$ using possibly random functions $u$ and $v$.  Then, we have $\calR(\probP^0_U, \probP^1_U) \supset \calR(\probP^0_V, \probP^1_V)$ if and only if there exists a conditional distribution $\probP_{V | U}$ such that for all $v$,
		\begin{equation*}
			\probP^i_V(v) = \sum_u \probP_{V|U}(v \mid u) \probP^i_U(u)
		\end{equation*}
		for $i = 0,1$, i.e., inclusion of $\calR$ is equivalent to channel degradation.
		\label{fact: blackwell}
	\end{fact}
	
	Note that the same conditional distribution $\probP_{V|U}$ must generate $\probP^i_V$ from $\probP^i_U$ for both $i = 0$ and 1. Hence, to show that a statistic $U$ is better than $V$, it suffices to construct a channel from $U$ to $V$ that maps $\probP^i_U$ to $\probP^i_V$ for both $i = 0,1$.

	A distributed version of this problem is as follows: the decision maker $\party{C}$ no longer observes $X$ and $Y$ directly. Instead, agents $\party{A}$ and $\party{B}$ respectively observe $X$ and $Y$, and then compute and send $g(X)$ and $h(Y)$ to $\party{C}$, where $g$ and $h$ are possibly random functions, and together referred to as a code for $(X,Y)$.
	The set of all possible type-I and type-II error probabilities of this distributed hypothesis test is given by the region $\calR\big(\probP^0_{g(X), h(Y)}, \probP^1_{g(X), h(Y)}\big)$. For a given code $(g,h)$, the optimal decision rules are given by the likelihood ratio test as earlier, but using the induced distributions $\probP^i_{g(X), h(Y)}$. 
	The code $(g,h)$ has uniformly better type-I and type-II error probabilities than $(g', h')$ if and only if $(g(X), h(Y)) \better (g'(X), h'(Y))$.

	\subsection{$n$ i.i.d. samples of correlated binary data}
	As described in Section~\ref{sec: intro}, we now suppose that the agents observe $n$ i.i.d.\ samples drawn from a doubly symmetric binary source whose correlation is determined by the hypothesis. 
	We write $(X,Y) \sim \DSBS(p_i)$ to mean $X$ and $Y$ are $\Bern(\frac12)$ and $\Pr\{X \neq Y\} = p_i$. The correlation between $X$ and $Y$ that are $\DSBS(p)$ is given by 
	\begin{equation*}
		\rho_i = \frac{\bbE[XY] - \bbE[X]\bbE[Y]}{\sqrt{\var(X)\var(Y)}} = 1-2p_i,
	\end{equation*}
	and hence, we have positive correlation if $p_i < \frac12$, negative correlation if $p_i > \frac12$ and independence otherwise. Also note that $\rho_1 = - \rho_0$ if and only if $p_1 = 1 - p_0$.
	Hence, under hypothesis $\hyp = i$, we have that $(X_\ell, Y_\ell)$, $\ell = 1,\dots,n$ is independently drawn according to $\DSBS(p_i)$, i.e., $\probP^i_{X^nY^n} = \DSBS(p_i)^{\otimes n}$.
	In the co-located version where $\party{C}$ observes both $X^n$ and $Y^n$, the modulo 2 sum $X^n \oplus Y^n \coloneqq (X_\ell \oplus Y_\ell)_{\ell=1}^n $ is a sufficient statistic of $(X^n, Y^n)$ (in fact, the Hamming weight of $X^n \oplus Y^n$ is sufficient).
	
	In the distributed setup, $\party{C}$ observes only the coded observations $g(X^n)$ and $h(Y^n)$, for some functions $g : \{0,1\}^n \to \{0,1\}^{k_X}$ and $h : \{0,1\}^n \to \{0,1\}^{k_Y}$, with $k_X, k_Y \leq n$. Let $|g| \coloneqq |\{g(x^n) : x^n \in \{0,1\}^n\}|$ denote the number of possible outputs of the function $g$, then we define the \emph{rate} of $g$ as the number of output bits per input bit, i.e., $R_X = \frac{1}{n}\log_2 |g|$, which is at most $k_X/n$.  Similarly, we define $R_Y = \frac{1}{n}\log_2 |h|$, which is at most $k_Y/n$.

	In particular, we study \emph{linear codes}, where $g$ and $h$ be linear functions, i.e., $G$ and $H$ are matrices of dimensions $k_X \times n$ and $k_Y \times n$ respectively, and $g(X^n) = G X^n$ and $h(Y^n) = H Y^n$, with matrix multiplication modulo 2, and $X^n$ and $Y^n$ taken as column vectors.  Without loss of generality, we may assume that the matrices are full-row-rank, as we can always convert a rank-deficient $k \times n$ matrix $G$ with $k \leq n$ to a $\rank(G) \times n$ full-rank matrix by a one-to-one transformation.  Hence, the rates of $g$ and $h$ are respectively given by $k_X/n$ and $k_Y/n$.

	Our motivation for studying linear codes is the following: Since $Z^n \eqqcolon X^n \oplus Y^n$ is a sufficient statistic in the co-located problem, one possible approach is to choose $g$ and $h$ so that $\party{C}$ is able to compute an estimate $\hat Z^n$ that is ``close'' to $Z^n$. K\"orner and Marton~\cite{korner_marton} showed that for computing the modulo 2 sum of correlated binary vectors in such a distributed fashion, linear codes can be used to obtain the same performance as classical random quantization but at a lower sum-rate.   This suggests that the K\"orner--Marton coding scheme might also be well-suited to this DHT problem, as is indeed shown by Haim and Kochman~\cite{haim_kochman_itw} when $p_0$ and $p_1$ are both smaller than $\frac12$.

	Before proceeding to our technical results, it will be convenient to represent $(X,Y) \sim \DSBS(p_i)$, $i=0,1$ using the following equivalent notation. Assume, without loss of generality (by re-labelling the hypotheses, if necessary) that $p_1$ is closer to $1/2$ than $p_0$ is (i.e., $|p_0 - 1/2| \geq |p_1 - 1/2|$).
	Let $X \sim \Bern(\frac12)$, $Z \sim \Bern(p_0)$ and $\tilde Z \sim \Bern(\tilde p)$ be independent, with $\tilde p$ such that $p_1 = p_0 * \tilde p \coloneqq p_0(1-\tilde p) + \tilde p (1-p_0)$---such a $\tilde p$ is guaranteed to exist as $|p_0-\frac12| \geq |p_1 - \frac12|$. Then, we clearly have $Y = X \oplus Z \oplus i \tilde Z$.

	In particular, for $p_1 = 1 - p_0$, we have $\tilde p = 1$, and hence $Y = X \oplus Z \oplus i$. Likewise, for $p_0 \neq p_1 = \frac12$, we also have $\tilde p = \frac12$.
	For the case of $n$ i.i.d.\ samples, we have 
	\begin{equation}
		Y^n = X^n \oplus Z^n \oplus i \tilde Z^n, \label{eqn: notation}
	\end{equation}
	where $X_\ell$ are i.i.d.\ $\Bern(\frac12)$, $Z_\ell$ are i.i.d.\ $\Bern(p_0)$ and $\tilde Z_\ell$ are i.i.d.\ $\Bern(\tilde p)$, and these are also independent of each other, with $\tilde p$ such that $p_1 = p_0 * \tilde p$.  Note that $p_0$ and $p_1$ are on opposite sides of $1/2$ if and only if $\tilde p \geq 1/2$.
	
	% If, on the other hand, $p_0$ is closer to $1/2$ than $p_1$ is, then we would have $Y^n = X^n \oplus Z^n \oplus (i-1) \tilde Z^n$, with $Z_\ell$ i.i.d.\ $\Bern(p_1)$ and $\tilde Z_\ell$ i.i.d.\ $\Bern(\tilde p)$, where $\tilde p$ is such that $p_0 = \tilde p * p_1$.
	
	\section{Best choices of linear codes} \label{sec: linear}
	Recall the partial order $\better$ on statistics defined in Section~\ref{sec: prelims}.  
	We now use the characterization in Fact~\ref{fact: blackwell} to obtain order relations involving statistics computed using linear codes.  First, we show that it is always better to use the same linear code to compress $X^n$ and $Y^n$, and that it is sufficient to base the decision rule on the modulo 2 sum of the compressed messages.  
	That is, the optimal decision rule based on the statistic $G(X^n \oplus Y^n)$ has uniformly better type-I and type-II error probabilities than the optimal rule based on the statistic $(GX^n, HY^n)$.  Similarly, the optimal rule based on $H(X^n \oplus Y^n)$ also has uniformly better type-I and type-II error probabilities than the optimal rule based on the statistic $(GX^n, HY^n)$.
	
	% 1. for any $p,q$, both using $G$ better than $(G,H)$. 
	\begin{theorem}[Same linear code is (always) better]
		For any $p_0, p_1 \in [0,1]$, consider the hypothesis test $\DSBS(p_0)$ versus $\DSBS(p_1)$ using $n$ i.i.d.\ samples. For any $n$-column matrices $G,H$, we have $G(X^n \oplus Y^n) \better (GX^n, HY^n)$ and $H(X^n \oplus Y^n) \better (GX^n, HY^n)$.   \label{thm: same}
	\end{theorem}
	\begin{proof}
		We prove the statement by constructing a channel to simulate $(GX^n, HY^n)$ from $G(X^n \oplus Y^n) = G(Z^n \oplus i \tilde Z^n)$ under $\hyp = i$.  Let $Y'_\ell$ be i.i.d.\ $\Bern(\frac12)$, independent of $X^n,Y^n$.  Set $U = G(X^n \oplus Y^n)$, $V_1 = U \oplus GY'^n$ and $V_2 = H Y'^n$, then $(V_1,V_2)$ has the same distribution as $(GX^n, HY^n)$ and we are done by Fact~\ref{fact: blackwell}.  The other side holds similarly.
	\end{proof}

	It is worth noting that this is not true if non-linear codes are allowed. For example, if we take $n = 2$, $g(x^2) = x_1x_2$ and $h(x^2) = \max\{x_1,x_2\}$ (these are the logical AND and OR operations on the two bits), then $(g(X^2), h(Y^2))$ and $(g(X^2), g(Y^2))$ cannot be compared by $\better$.

	We say that a code $(g,h)$ is the \emph{best} in some class $\calC$ of codes if $(g(X^n), h(Y^n) \better (g'(X^n), h'(Y^n))$ for all $(g',h')$ in $\calC$.  It is possible that a given class $\calC$ may not have such a greatest element.  We show that when $\calC$ is the class of all linear codes of a given rate, there is indeed a greatest element for some choices of $(p_0, p_1)$. 
	
	First, note that by Theorem~\ref{thm: same}, when looking for the best linear codes, we may assume without loss of generality that $G = H$, i.e., both agents use the same full-row-rank $k \times n$ matrix $G$ to encode $X^n$ and $Y^n$ respectively to $GX^n$ and $GY^n$.  We now show that for certain values of $(p_0, p_1)$, the best linear code at rate $k/n$ is the most straightforward code possible---simply truncating the observation $X^n$ to $X^k$ by picking the matrix $G = [\eye_k \,  \zeros_{k \times (n-k)}]$ (more generally, any one-to-one transformation of any $k$ of $X_1,\dots,X_n$ will do). 
	This is in analogy with the dictator functions maximizing the mutual information of binary vectors \cite{kumar_courtade,dictator_mi},  
	though there does not appear to be a technical connection (indeed the ``multivariate'' generalization is known to be false~\cite{dictator_ctrex}). The values $(p_0, p_1)$ for which we are able to show this result are: (1) $p_0 = 1/2$ (testing for independence, $\rho_0 = 0$), (2) $p_1 = 1/2$ (testing against independence, $\rho_1 = 0$), and (3) $p_1 = 1 - p_0$ (testing positive against negative values of the same magnitude of correlation, $\rho_1 = - \rho_0$).
	
	% 2. for $q = 1/2, 1-p$, truncation is best $G$.
	\begin{theorem}[Truncation is (sometimes) the best linear code]
		For any $p \in [0,1]$, consider the following hypothesis tests:
		\begin{itemize}
			\item[(1)] $\DSBS(p)$ versus $\DSBS(\frac12)$,
			\item[(2)] $\DSBS(\frac12)$ versus $\DSBS(p)$, and
			\item[(3)] $\DSBS(p)$ versus $\DSBS(1-p)$
		\end{itemize}
		using $n$ i.i.d.\ samples. 
		For each test above, for any $n$-column matrices $G,H$, we have $(X^k, Y^k) \better (GX^n, HY^n)$ where $k = \min\{\rank(G), \rank(H)\}$.  \label{thm: trunc}
	\end{theorem}
	\begin{proof} 
		In all cases, it suffices, by Fact~\ref{fact: blackwell}, to construct a channel to simulate $G(X^n \oplus Y^n) = G(Z^n \oplus i \tilde Z^n)$ from $X^k \oplus Y^k = Z^k \oplus i \tilde Z^k$ for any $k$-rank $G$, as Theorem~\ref{thm: same} already gives us $G(X^n \oplus Y^n) \better (GX^n, HY^n)$.
		We may assume without loss of generality that $G$ has full row-rank and $G = [\eye_{k} \ A]$ for some $k \times (n-k)$ matrix $A$ (if not, replace $k$ by $\rank(G)$ and use row-reductions to obtain $\eye_k$).  
		Note that (1) and (2) are equivalent, and to use \eqref{eqn: notation}, we need $p_1$ to be closer to $1/2$.  Hence both (1) and (2) are proved by considering $p_1 = 1/2$ in \emph{a} below, and (3) is proved in \emph{b}. 
		
		\noindent \paragraph{$p_1 = \frac12$ } Note that $G(X^n \oplus Y^n) = G(Z^n \oplus i \tilde Z^n) = Z^k \oplus A Z_{k+1}^n \oplus i G \tilde Z^n$, with $\tilde Z_\ell$ being i.i.d.\ $\Bern(\frac12)$.    Since the modulo 2 sum of $\Bern(\frac12)$ random variables is still $\Bern(\frac12)$, we have that the distribution of $G \tilde Z^n$ is the same as that of $\tilde Z^k$.  Let $Z'_\ell$ be i.i.d.\ $\Bern(p)$, independent of everything else, then letting $U = Z^k \oplus i \tilde Z^k$ and $V = U \oplus A{Z'}_{k+1}^n$, we have that $V$ has the same distribution as $G(Z^n \oplus i \tilde Z^n)$, which completes the proof.
		
		The case with $p_0 = \frac12$ follows similarly.
		
		\noindent \paragraph{$p_1 = 1-p_0$. } Classify the rows of $A$ as having even versus odd weight, and assume, again without loss of generality, that $A = \begin{bmatrix} A_e \\ A_o \end{bmatrix}$, where $A_e$ has even-weight rows and $A_o$ has odd-weight rows.  Let $k_e$ be the number of rows of $A_e$, and $k_o$ the number of rows of $A_o$.  Then $G$ is of the form
		$G = [B\ C]$
		with
		$B = \begin{bmatrix} \eye_{k_e} \\ \zeros \end{bmatrix}$
		and
		$C = \begin{bmatrix} \zeros & A_e \\ \eye_{k_o} & A_o \end{bmatrix}$.
		Observe that all rows of $C$ are of even weight.
		We now show that $X^{k_e} \oplus Y^{k_e} \better G(X^n \oplus Y^n)$, then as $X^k \oplus Y^k \better X^{k_e} \oplus Y^{k_e}$ trivially, we are done.
		
		As $p_1 = 1 - p_0$, we have $\tilde p = 1$, i.e., $\tilde Z_\ell = 1$ for all $i$. With $G = [B\ C]$, we have
		$G(Z^n \oplus i \ones^n) = B (Z^{k_e} \oplus i \ones^{k_e}) \oplus W$,
		where $W = C (Z_{k_e+1}^n \oplus i \ones_{k_e+1}^n)$. Since all rows of $C$ are of even weight, $C \ones_{k_e+1}^n = \zeros_{k_e+1}^n$, and hence, $W = C Z_{k_e+1}^n$.
		Let $Z'_\ell$ be i.i.d.\ $\Bern(p)$, independent of everything else, and let $U = Z^{k_e} \oplus i Z^{k_e}$ and $V = \begin{bmatrix} U  \\ \zeros \end{bmatrix} + C {Z'}_{k_e+1}^n$.  Then $V$ has the same distribution as $G(X^n \oplus Y^n)$, and this completes the proof.
	\end{proof}
	
	Thus, we see that truncation is the best linear code for some values of $(p_0, p_1)$, given by the thick lines in Fig.~\ref{fig: p0p1}.  We conjecture that this holds more generally, whenever $(p_0, p_1)$ are on opposite sides of $1/2$, i.e., $\rho_0 \rho_1 \leq 0$, marked by the shaded region in Fig.~\ref{fig: p0p1}.
	
	\begin{conjecture}
		For any $p_0 \neq p_1$,  such that $(p_0-1/2)(p_1-1/2) \leq 0$, consider the hypothesis test $\DSBS(p_0)$ versus $\DSBS(p_1)$ using $n$ i.i.d.\ samples. For any $n$-column matrices $G,H$, we have $(X^k, Y^k) \better (GX^n, GY^n)$ where $k = \min\{\rank(G), \rank(H)\}$.  \label{conj: trunc}
	\end{conjecture}
	
	\begin{figure}[t]
		\centering
		        \begin{tikzpicture}[scale=2.8]
          % Axes
          \draw[->] (0,0) -- (1.1,0) node[right] {$p_0$};
          \draw[->] (0,0) -- (0,1.1) node[above] {$p_1$};
          
          % Grid box
          \draw[black, thin] (0,0) rectangle (1,1);
          
          % Shaded regions: p0 < 1/2, p1 > 1/2  and  p0 > 1/2, p1 < 1/2
          \fill[pattern color = red, pattern=north west lines] (0,0.5) -- (0.5,0.5) -- (0.5,1) -- (0,1) -- cycle;
          \fill[pattern color = red, pattern=north west lines] (0.5,0) -- (1,0) -- (1,0.5) -- (0.5,0.5) -- cycle;
          
          % Lines
          \draw[very thick] (0.5,0) -- (0.5,1);
          \draw[very thick] (0,0.5) -- (1,0.5);
          \draw[very thick] (0,1) -- (1,0);
          
          % Ticks and labels
          \foreach \x in {0,0.5,1}
            \draw (\x,0) -- (\x,-0.015) node[below] {\scriptsize $\x$};
          \foreach \y in {0,0.5,1}
            \draw (0,\y) -- (-0.015,\y) node[left] {\scriptsize $\y$};
        \end{tikzpicture}
		\caption{$(p_0, p_1)$-plane where truncation is the best linear scheme for testing $\DSBS(p_0)$ versus $\DSBS(p_1)$: thick, black lines where it is known (Theorem~\ref{thm: trunc}) and red, shaded region where it is conjectured (Conjecture~\ref{conj: trunc}).}
		\label{fig: p0p1}
	\end{figure}
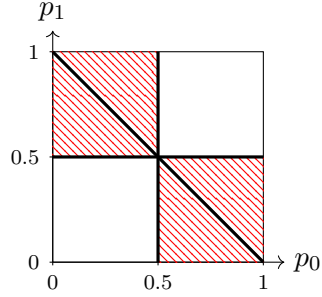
	
	Note that $p_0$ and $p_1$ are on opposite sides of $\frac12$ if and only if $\tilde p$ is at least $\frac12$.  Theorem~\ref{thm: trunc} shows that truncation is the best linear scheme for the two extreme cases of $\tilde p = \frac12$ and $1$, but the channels that we construct to show the degradation is different for both cases.  Unifying these proofs might lead to a proof for Conjecture~\ref{conj: trunc}. We have been unsuccessful at finding an explicit channel in the general case, which makes us think that the proof might require more complicated machinery as developed for boolean functions~\cite{ODonnell_2014}.
	Though we are unable to prove Conjecture~\ref{conj: trunc} here, we provide some empirical support for it in Section~\ref{sec: emp_evidence}.

	\section{Comparison with random quantization}
	
	Given that truncation is the best linear code for a set of $(p_0,p_1)$ pairs, it is now easy to check that a linear code is suboptimal for such $(p_0, p_1)$---simply find a random code that does better than truncation and we are done.
	It is usually difficult to evaluate the exact type-I and type-II error probabilities of random codes, so we focus on the asymptotic Stein regime, as is common in the DHT literature. That is, we take $n \to \infty$, and consider the best achievable type-II error \emph{exponent} $\lim_{n \to \infty} -\frac{1}{n}\log e_2$ while requiring that the type-I error exponent is smaller than some constant $\epsilon \in (0,1)$.  We denote this best achievable type-II exponent when both agents operate at rate $R$ by $E(R)$.  In the co-located case, the Stein exponent is $E(\infty) = D(\probP^0_{XY} \,\|\, \probP^1_{XY})$.
	In the binary setting that we consider with $\probP^i_{XY} = \DSBS(p_i)$, we have $E(\infty) = E(1) = \bindiv(p_0 \,\|\, p_1)$, i.e., $e_2 = \exp(- n \bindiv(p_0 \,\|\, p_1) + \o(n))$, where $\bindiv(a \,\|\, b)$ denotes the KL divergence between $\Bern(a)$ and $\Bern(b)$ distributions.
	
	Hence, in the distributed case using truncation to encode $(X^n, Y^n)$ to $(X^k, Y^k)$, the best achievable type-II error is $e_2 = \exp(- k \bindiv(p_0 \,\|\, p_1) + \o(k))$. The Stein exponent is then simply
	\begin{equation*}
		\truncE(R) = \lim_{n \to \infty} -\tfrac{1}{n} (- k \bindiv(p_0 \,\|\, p_1) + \o(k)) = R \bindiv(p_0 \,\|\, p_1),
	\end{equation*}
	or equivalently, the exponent-rate trade-off is given by a straight line joining $(0,0)$ to $(1, \bindiv(p_0 \,\|\, p_1))$ as shown in Fig.~\ref{fig: trunc_vs_random}.

	We now consider Han's random coding scheme~\cite{han_multiterminal}.  It requires the agents to pick ``test channels'' $\probP_{U|X}$ and $\probP_{V|Y}$, which allows them to communicate with the decision maker at rates $I(U;X)$ and $I(V;Y)$ respectively (see~\cite{han_amari_survey,kochman_wang_improved} for a detailed description of the scheme).  
	Given these test channels, let  $\probP^i_{UXYV}$ denote the induced joint distribution given by $\probP^i_{XY}\probP_{U|X}\probP_{V|Y}$
	The Stein exponent achieved by Han's random coding scheme is given by
	\begin{equation}
		\HanE(\probP_{U|X}, \probP_{V|Y}) = \min_{\pi_{UXYV}} D( \pi_{UXYV} \,\|\, \probP^1_{UXYV}) \label{eqn: han_exp}
	\end{equation}
	with the minimum taken over distributions $\pi_{UXYV}$ having the same $UX$, $UY$, and $UV$ marginals as $\probP^0_{UXYV}$.   
	By the theory of I-projections~\cite{csiszar_shields}, the optimal $\pi^*$ in the above minimization has the following structure:
	\begin{equation*}
		\pi^*(u,x,y,v) = \probP^1_{UXYV}(u,x,y,v) f_1(u,x) f_2(v,y) f_3(u,v)
	\end{equation*}
	for some functions $f_1, f_2, f_3$.  This gives us the Markov relations $X \markov (U,V) \markov Y$ and $U \markov (X,Y) \markov V$ under $\pi^*$.  For testing against independence, i.e., when $\probP^1_{XY} = \probP^0_{X}\probP^0_{Y}$, this simplifies to $X \markov U \markov V \markov Y$.  
	
	For our DSBS setup, since we consider symmetric rate constraints, we choose $\probP_{U|X} = \probP_{V|Y}$ to both be binary symmetric channels (BSC) of the same parameter, say $\alpha \in (0,\frac12)$.  To operate at rate $R$, we require $I(U;X) = 1 - \binent(\alpha) = R$, where $\binent(a)$ denotes the entropy of $\Bern(a)$.  This is not known to be the optimal choice of test channel, but experiments suggest that this is indeed the case for any $(p_0, p_1)$ on opposite sides of $1/2$. 
	In Fig.~\ref{fig: trunc_vs_random}, we plot $\HanE$ for a representative $(p_0, p_1)$ by numerically solving \eqref{eqn: han_exp} for different choices of $\probP_{U|X}$ and $\probP_{V|Y}$ that are randomly picked from a Dirichlet distribution.  Our choice $(p_0, p_1) = (0.1, 0.9)$ satisfies $p_1 = 1 - p_0$, since it is only in this case (and testing for/against independence) where we know that truncation is better than other linear codes, letting us meaningfully compare linear and random codes via truncation, but the plots are qualitatively identical for other $p_0, p_1$ on opposite sides of $1/2$ as well. 
	
	Let $\HanE(R)$ denote the Han exponent on taking $\probP_{U|X}$ and $\probP_{V|Y}$ to be the same BSC at rate $R$. 
	We see that $R \mapsto \HanE(R)$ is not concave and there is a value of $R$ such that $\HanE(R) > \truncE(R)$.  This allows us to improve on both $\HanE$ and $\truncE$ in a trivial manner by appropriately truncating the quantized vectors $(U^n, V^n)$, which yields an improved random coding exponent $\comE$ given by the upper concave envelope of $\HanE$ (Fig.~\ref{fig: trunc_vs_random}). First generate a random code with $\alpha = \alpha^*$ so that we are at the point on the $\HanE(R)$ curve where the tangent from the origin touches it, and then truncate to obtain all points on the line to the origin. 
	
	\begin{figure}[t]
		\centering
		\includegraphics[width=0.68\linewidth]{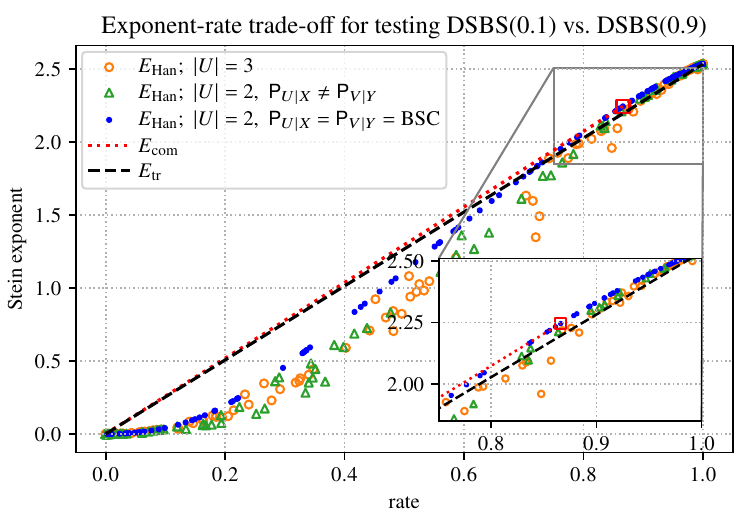} 
		\caption{Comparison of $\truncE$ and $\HanE$ for $(p_0, p_1) = (0.1, 0.9)$.  Choosing  $\probP_{U|X}$ and $\probP_{V|Y}$ to be the same BSC seems to be the best choice of test channels in Han's scheme.  At sufficiently large rates, $\HanE(R) > \truncE(R)$, and hence we have an improved exponent $\comE(R)$ which is strictly better than $\truncE(R)$ at all $R$ and strictly better than $\HanE(R)$ for small $R$.}
		\label{fig: trunc_vs_random}
	\end{figure}

	For testing against independence in particular, i.e., $p_1 = 1/2$, we can arrive at this conclusion analytically.
	Since $\pi^*$ has the same $UX, UV, VY$ marginals $\probP^0$, we can identify the transition probabilities at each step of the $X \markov U \markov V \markov Y$ Markov chain as $\pi^*_{X|U} = \BSC(\alpha)$, $\pi^*_{V|U} = \BSC(\alpha * \alpha * p)$, and $\pi^*_{Y|V} = \BSC(\alpha)$.  We can then explicitly compute Han's exponent as $\HanE = I(U;V) = 1 - \binent(\alpha * \alpha * p)$ at rate $R = I(U;X) = 1 - \binent(\alpha)$.  This gives us the exponent-rate trade-off curve parametrized by $\alpha \in (0,1/2)$.  Note that $\HanE(1) = \truncE(1) = \bindiv(p_0 \,\|\, \frac12)$.  It is also easy to see that the derivative $\HanE'(1) = 0$ by checking that
	\begin{equation*}
		\lim_{\alpha \to \frac12} \frac{ \frac{\d }{\d \alpha}1 - \binent(\alpha * \alpha * p)  }{  \frac{\d }{\d \alpha}1 - \binent(\alpha) } = 0,
	\end{equation*}
	and hence, at a sufficiently large rate, Han's quantization scheme achieves a strictly higher Stein exponent than truncation.

	It might be of interest to note that for testing against independence with a rate-constraint only on $X^n$ (i.e., $R_X \leq R$, $R_Y \leq \infty$), this quantization is known to be optimal~\cite{ahlswede_csiszar}.  The exponent-rate pair is $(I(U;Y), I(U;X)) = (1 - \binent(\alpha * p), 1- \binent(\alpha))$; this curve is concave and always lies above the truncation exponent given by the straight line (also called the $F_I$-curve~\cite{witsenhausen_wyner,calmon_sdpi} or the information bottleneck function).  In our setting with a symmetric rate constraint, the exponent-rate pair is given by a different ``bottleneck'' $(I(U;V), I(U;X)) = (1 - \binent(\alpha * \alpha * p), 1- \binent(\alpha))$, which turns out to not be concave.

	\section{Numerical support for Conjecture~\ref{conj: trunc}} \label{sec: emp_evidence}
	
	\begin{figure}[t]
		\centering
		\includegraphics[width=0.61\linewidth]{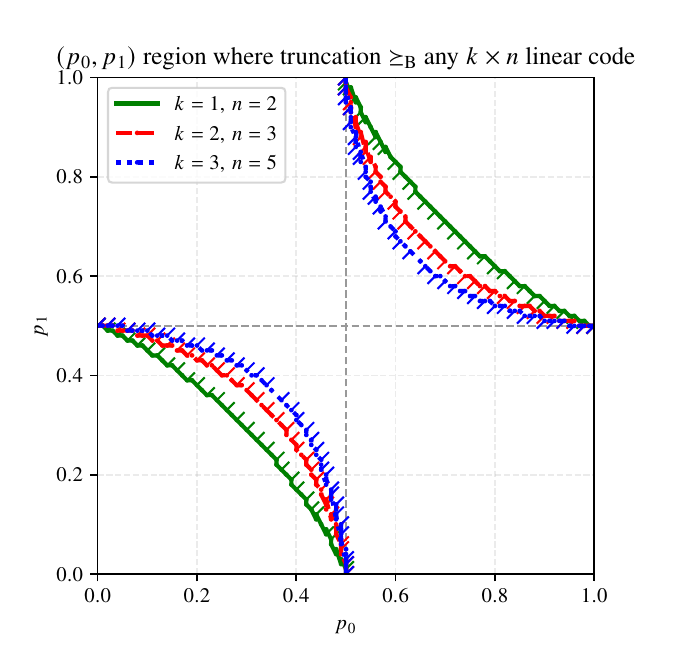} \vspace*{-5pt}
		\caption{The hatched region contained between the two curves shows the $(p_0,p_1)$ region where $X^k \oplus Y^k \better G(X^n \oplus Y^n)$ for all rank-$k$, $n$-column matrices $G$, shown for three values of $(k,n) = (1,2),\,(2,3),\,(3,5)$. The region appears to be shrinking to $(p_0 - 1/2)(p_1 - 1/2) \leq 0$ as $n \to \infty$.} \vspace*{-8pt}
		\label{fig: finite_n}
	\end{figure}
	
	In Section~\ref{sec: linear}, we conjectured that truncation is the best linear scheme for testing $\DSBS(p_0)$ versus $\DSBS(p_1)$ for any $p_0$ and $p_1$ on opposite sides of $1/2$.  By Fact~\ref{fact: blackwell}, to show that truncation is better than a full-row-rank $k \times n$ matrix $G$,
	it suffices to show the existence of a $2^k \times 2^k$ row-stochastic matrix $\Wmat$ such that $\bvec_i = \Wmat \avec_i$ for $i = 0,1$, where $\bvec_i$ has entries given by $\probP^i_{G(X^n \oplus Y^n)}(s^k)$ and $\avec_i$ has entries given by $\probP^i_{X^k \oplus Y^k}(s^k)$, with $s^k$ going over all vectors in $\{0,1\}^k$.  This is a linear feasibility problem, which can be solved numerically for small $(k,n)$. In Fig.~\ref{fig: finite_n}, for each choice of $(k,n)$, we go over all $(p_0,p_1)$ pairs (discretized to a $0.01 \times 0.01$ grid) and check whether or not a $\Wmat$ as above exists for each $G$.  For $(k,n) = (1,2),\,(2,3)$, we can explicitly compute these boundaries to be the rectangular hyperbolae $2p_0p_1=1$ and $p_0+p_1=3p_0p_1$ (and their reflections about $p_0 + p_1 = 1$).
	The region appears to be shrinking as $n$ increases, which leaves us with another conjecture: this region shrinks to exactly the set of $(p_0, p_1)$ on opposite sides of $1/2$ as marked in Fig.~\ref{fig: p0p1}  as $n \to \infty$.

	\printbibliography

@inproceedings{berger1979decentralized,
  title={Decentralized estimation and decision theory},
  author={Berger, Toby},
  booktitle={IEEE Seven Springs Workshop on Information Theory},
  volume={120},
  year={1979},
}

@ARTICLE{ahlswede_csiszar,
  author={Ahlswede, R. and Csiszar, I.},
  journal={IEEE Transactions on Information Theory}, 
  title={Hypothesis testing with communication constraints}, 
  year={1986},
  volume={32},
  number={4},
  pages={533-542},
  doi={10.1109/TIT.1986.1057194}
}

@ARTICLE{han_multiterminal,
  author={Te Sun Han},
  journal={IEEE Transactions on Information Theory}, 
  title={Hypothesis testing with multiterminal data compression}, 
  year={1987},
  volume={33},
  number={6},
  pages={759-772},
  doi={10.1109/TIT.1987.1057383}
}

@ARTICLE{han_amari_survey,
  author={Te Sun Han and Amari, S.},
  journal={IEEE Transactions on Information Theory}, 
  title={Statistical inference under multiterminal data compression}, 
  year={1998},
  volume={44},
  number={6},
  pages={2300-2324},
  doi={10.1109/18.720540}
}

@INPROCEEDINGS{haim_kochman_itw,
  author={Haim, Eli and Kochman, Yuval},
  booktitle={Proc.\ IEEE Information Theory Workshop (ITW)}, 
  title={Binary distributed hypothesis testing via {K}örner-{M}arton coding}, 
  year={2016},
  volume={},
  number={},
  pages={146-150},
  doi={10.1109/ITW.2016.7606813}
}

@misc{haim_kochman_arxiv,
      title={On Binary distributed hypothesis testing}, 
      author={Eli Haim and Yuval Kochman},
      year={2018},
      eprint={1801.00310},
      archivePrefix={arXiv},
      primaryClass={cs.IT},
      url={https://arxiv.org/abs/1801.00310}, 
}

@ARTICLE{korner_marton,
  author={K\"orner, J. and Marton, K.},
  journal={IEEE Transactions on Information Theory}, 
  title={How to encode the modulo-two sum of binary sources (Corresp.)}, 
  year={1979},
  volume={25},
  number={2},
  pages={219-221},
  doi={10.1109/TIT.1979.1056022}
}

@inproceedings{blackwell_order,
  author       = {David Blackwell},
  title        = {Comparison of experiments},
  booktitle    = {Proc.\ Second Berkeley Symposium on Mathematical Statistics and Probability},
  year         = {1951},
  pages        = {93--102},
}

@book{Polyanskiy_Wu_2025, place={Cambridge}, title={Information Theory: From Coding to Learning}, publisher={Cambridge University Press}, author={Polyanskiy, Yury and Wu, Yihong}, year={2025}
}

@INPROCEEDINGS{practical_DHT,
  author={Dupraz, Elsa and Adamou, Ismaila Salihou and Asvadi, Reza and Matsumoto, Tad},
  booktitle={Proc.\ IEEE International Symposium on Information Theory (ISIT)}, 
  title={Practical short-length coding schemes for binary distributed hypothesis testing}, 
  year={2024},
  volume={},
  number={},
  pages={2915-2920},
  doi={10.1109/ISIT57864.2024.10619545}
}

@ARTICLE{kumar_courtade,
  author={Courtade, Thomas A. and Kumar, Gowtham R.},
  journal={IEEE Transactions on Information Theory}, 
  title={Which boolean functions maximize mutual information on noisy inputs?}, 
  year={2014},
  volume={60},
  number={8},
  pages={4515-4525},
  doi={10.1109/TIT.2014.2326877}
}

@article{dictator_mi,
author = {Georg Pichler and Pablo Piantanida and Gerald Matz},
title = {Dictator functions maximize mutual information},
volume = {28},
journal = {The Annals of Applied Probability},
number = {5},
publisher = {Institute of Mathematical Statistics},
pages = {3094-3101},
year = {2018},
doi = {10.1214/18-AAP1384},
}

@inproceedings{dictator_ctrex,
  title={Most informative quantization functions},
  author={Chandar, Venkat and Tchamkerten, Aslan},
  booktitle={Proc.\ Information Theory and Applications (ITA) Workshop},
  year={2014}
}

@ARTICLE{kochman_wang_improved,
  author={Kochman, Yuval and Wang, Ligong},
  journal={IEEE Transactions on Information Theory}, 
  title={Improved random-binning exponent for distributed hypothesis testing}, 
  year={2025},
  volume={71},
  number={11},
  pages={8217-8222},
  doi={10.1109/TIT.2025.3603269}
}

@article{csiszar_shields,
  title={Information theory and statistics: A tutorial},
  author={Csisz{\'a}r, Imre and Shields, Paul C},
  journal={Foundations and Trends{\textregistered} in Communications and Information Theory},
  volume={1},
  number={4},
  pages={417--528},
  year={2004},
  publisher={Now Publishers, Inc.}
}

@ARTICLE{witsenhausen_wyner,
  author={Witsenhausen, H. and Wyner, A.},
  journal={IEEE Transactions on Information Theory}, 
  title={A conditional entropy bound for a pair of discrete random variables}, 
  year={1975},
  volume={21},
  number={5},
  pages={493-501},
  doi={10.1109/TIT.1975.1055437}
}

@ARTICLE{calmon_sdpi,
  author={Calmon, Flavio du Pin and Polyanskiy, Yury and Wu, Yihong},
  journal={IEEE Transactions on Information Theory}, 
  title={Strong data processing inequalities for input constrained additive noise channels}, 
  year={2018},
  volume={64},
  number={3},
  pages={1879-1892},
  doi={10.1109/TIT.2017.2782359}
}

@book{ODonnell_2014, place={Cambridge}, title={Analysis of Boolean Functions}, publisher={Cambridge University Press}, author={O’Donnell, Ryan}, year={2014}}

@article{sensor_networks,
title = {Secure data aggregation in wireless sensor networks: A comprehensive overview},
journal = {Computer Networks},
volume = {53},
number = {12},
pages = {2022-2037},
year = {2009},
issn = {1389-1286},
doi = {https://doi.org/10.1016/j.comnet.2009.02.023},
author = {Suat Ozdemir and Yang Xiao},
}

@ARTICLE{sensor_new,
  author={Trigka, Maria and Dritsas, Elias},
  journal={IEEE Access}, 
  title={Wireless sensor networks: From fundamentals and applications to innovations and future trends}, 
  year={2025},
  volume={13},
  number={},
  pages={96365-96399},
  doi={10.1109/ACCESS.2025.3572328}
}

@INPROCEEDINGS{DHT_interactive,
  author={Xiang, Yu and Kim, Young-Han},
  booktitle={50th Annual Allerton Conference on Communication, Control, and Computing (Allerton)}, 
  title={Interactive hypothesis testing with communication constraints}, 
  year={2012},
  volume={},
  number={},
  pages={1065-1072},
  doi={10.1109/Allerton.2012.6483336}
}

@article{DHT_cooperation,
  author={Escamilla, Pierre and Wigger, Michèle and Zaidi, Abdellatif},
  journal={IEEE Transactions on Information Theory}, 
  title={Distributed hypothesis testing: Cooperation and concurrent detection}, 
  year={2020},
  volume={66},
  number={12},
  pages={7550-7564},
  doi={10.1109/TIT.2020.3019654}
}

@inproceedings{DHT_privacy,
author = {Liao, Jiachun and Sankar, Lalitha and Calmon, Flavio du Pin and Tan, Vincent Y. F.},
title = {Hypothesis testing under maximal leakage privacy constraints},
year = {2017},
doi = {10.1109/ISIT.2017.8006634},
booktitle = {Proc.\ IEEE International Symposium on Information Theory (ISIT)},
pages = {779–783},
numpages = {5},
location = {Aachen, Germany}
}

@ARTICLE{DHT_noisy,
  author={Sreekumar, Sreejith and Gündüz, Deniz},
  journal={IEEE Transactions on Information Theory}, 
  title={Distributed hypothesis testing over discrete memoryless channels}, 
  year={2020},
  volume={66},
  number={4},
  pages={2044-2066},
  doi={10.1109/TIT.2019.2953750}
}

@INPROCEEDINGS{SHA,
  author={Shimokawa, H. and Te Sun Han and Amari, S.},
  booktitle={Proc.\ IEEE International Symposium on Information Theory (ISIT)}, 
  title={Error bound of hypothesis testing with data compression}, 
  year={1994},
  volume={},
  number={},
  pages={114},
  doi={10.1109/ISIT.1994.394874}
}

@INPROCEEDINGS{watanabe_binning,
  author={Watanabe, Shun},
  booktitle={Proc.\ IEEE International Symposium on Information Theory (ISIT)}, 
  title={On sub-optimality of random binning for distributed hypothesis testing}, 
  year={2022},
  volume={},
  number={},
  pages={2708-2713},
  doi={10.1109/ISIT50566.2022.9834275}}

@ARTICLE{amari_fisher,
  	author={Amari, Shun-ichi},
  	journal={IEEE Transactions on Information Theory}, 
  	title={On optimal data compression in multiterminal statistical inference}, 
  	year={2011},
  	volume={57},
  	number={9},
  	pages={5577-5587},
  	doi={10.1109/TIT.2011.2162270}}
	
\end{document}